\documentclass[letterpaper, 12pt]{llncs}

\usepackage{fullpage}
\usepackage{epsfig}
\usepackage{amssymb}
\usepackage{algorithmic}
\usepackage{natbib}
\usepackage{localVariables}

\usepackage{versions}
\excludeversion{REMOVED}

  {\begin{itemize}%
    \setlength{\itemsep}{\baselineskip}%
    \setlength{\parskip}{-\baselineskip}}%
  {\end{itemize}}

  {\begin{enumerate}%
    \setlength{\itemsep}{\baselineskip}%
    \setlength{\parskip}{-\baselineskip}}%
  {\end{enumerate}}

\makeatletter
\renewcommand\bibsection%
{
  \section*{\refname
    \@mkboth{\MakeUppercase{\refname}}{\MakeUppercase{\refname}}}
}
\makeatother

\begin{document}

\title{Succinct Geometric Indexes Supporting Point Location Queries}

\author{Prosenjit Bose\inst{1} \and Eric Y. Chen\inst{2} \and Meng He\inst{1} \and Anil Maheshwari\inst{1} \and Pat Morin\inst{1}}
\institute{
  School of Computer Science, 
  Carleton University, Canada, 
  \email\{jit, mhe, anil, morin\}@cg.scs.carleton.ca
\and
  Cheriton School of Computer Science,
  University of Waterloo, Canada,
  \email y28chen@uwaterloo.ca
}

\maketitle

\setcounter{page}{1}
\pagestyle{plain}
\thispagestyle{plain}

\begin{abstract}
We propose to design data structures called succinct geometric indexes of negligible space (more precisely, $o(n)$ bits) that, by taking advantage of the $n$ points in the data set permuted and stored elsewhere as a sequence, to support geometric queries in optimal time. 
Our first and main result is a succinct geometric index that can answer point location queries, a fundamental problem in computational geometry, on planar triangulations in $O(\lg n)$ time\footnote{We use $\lg n$ to denote $\lceil \log_2 n \rceil$.}. 
We also design three variants of this index. 
The first supports point location using $\lg n + 2\sqrt{\lg n} + O(\lg^{1/4} n)$ point-line comparisons. 
The second supports point location in $o(\lg n)$ time when the coordinates are integers bounded by $U$. 
The last variant can answer point location in $O(H+1)$ expected time, where $H$ is the entropy of the query distribution. 
These results match the query efficiency of previous point location structures that use $O(n)$ words or $O(n\lg n)$ bits, while saving drastic amounts of space. 

We then generalize our succinct geometric index to planar subdivisions, and design indexes for other types of queries. 
Finally, we apply our techniques to design the first implicit data structures that support point location in $O(\lg^2 n)$ time. 
\end{abstract}

\section{Introduction}
The problem of efficiently storing and retrieving geometric data sets that typically consist of collections of data points and regions is fundamental in computational geometry. 
Researchers have designed many data structures to represent geometric data, and to support various types of queries, such as point location~\cite{ki1983, egs1986, co1986, st1986}, nearest neighbour~\cite{in2004}, range searching~\cite{ae1999} and ray shooting~\cite{hs1993}. 

Among these queries, planar point location is perhaps the most fundamental and thus has been studied extensively. 
Given a planar subdivision, the problem is to construct a data structure so that the face of the subdivision containing a query point can be located quickly. 
In the 1980s, various researchers~\cite{ki1983, egs1986, co1986, st1986} showed that data structures of $O(n)$ words, where $n$ is the number of vertices of the planar subdivision, can be constructed to support point location in $O(\lg n)$ time, which is asymptotically optimal. 

Researchers have also considered improving the query efficiency of point location structures under various assumptions. 
Several researches~\cite{gor1997,sa2000} considered the exact number of steps (i.e. point-line comparisons) required to answer point locations queries. 
Seidel and Adamy~\cite{sa2000} showed that there is an $O(n)$-word structure that can answer point location in $\lg n + 2\sqrt{\lg n} + O(\lg^{1/4} n)$ steps. 
Researchers later considered the case where the query distribution is known. 
If the probability of the $i${\kth} face of the planar subdivision containing the query point is $p_i$, the lower bound of the expected time of answering a query under the binary decision tree model is the entropy $H = \sum_{i=1}^f(p_i\log_2 \frac{1}{p_i})$, where $f$ is the number of faces. 
When the planar subdivision is a planar triangulation, data structures of $O(n)$ words can be constructed to answer point location queries in $O(H+1)$~\cite{in2004} expected time or even using $H+O(\sqrt{H}+1)$~\cite{ammw2007} expected comparisons per query. 
Recently, Chan~\cite{ch2006} and P\v{a}tra\c{s}cu~\cite{pa2006} considered the case where the coordinates of the points are integers bounded by $U \le 2^w$, and proposed a linear space structure that answers point location queries in $O(\min\{\lg n/\lg\lg n, \sqrt{\lg U}\})$ time. 

As we have already seen, much work has been done to improve the query efficiency of point location. 
However, much less effort has been made to further reduce the storage cost. 
As a result of the rapid growth of geometric data sets available in Geometric Information Systems (GIS), spatial databases and graphics, many modern applications process geometric data measured in gigabytes or even terabytes. 
Although the above point location structures require linear space, the constants hidden in the asymptotic space bounds are usually large, so that they often occupy space many times the size of the geometric data. 
When the size of the data is huge, it is often impossible or at least undesirable to construct and store these data structures. 
Most data structures supporting other types of geometric queries are facing the same problem.

Some attempts, however, have been made to improve the space efficiency of various geometric data structures. 
Goodrich~\etal~\cite{gor1997} showed that given a planar triangulation, a structure of sublinear space can be constructed to answer point location queries in $O(\lg n)$ time. 
However, their approach assumed that the connectivity information (i.e. adjacencies) of the planar triangulation is given and stored elsewhere. 
This information can easily occupy much more space than that required to store the point coordinates (an adjacency list for a planar triangulation would take about 4n words), and can make the total space of the point location structure to be $O(n)$ words. 
By applying the idea of implicit data structures~\cite{M86}, researchers~\cite{BCC04, CC08} have designed some implicit geometric data structures. 
The idea is to store a permuted sequence of the point set, so that with zero or $O(1)$ extra space, geometric queries can be answered efficiently. 
The most recent result by Chan and Chen~\cite{CC08} showed that an implicit structure can be constructed to answer nearest neighbour query in the plane in $O(\lg^{1.71} n)$ time. 
This approach saves a lot of space,  but there are still limitations. 
First, the above query time is not asymptotically optimal. 
Second, it is not known how to support point location in planar triangulations or planar subdivisions using implicit data structures. 

Have researchers tried all the major known techniques to design space-efficient geometric data structures? The answer is no. 
There has been another line of research on data structures called {\em succinct data structures}. 
Succinct data structures\index{succinct data structure} were first proposed by
Jacobson~\cite{j1989} to encode bit
vectors, (unlabeled) trees and planar graphs in space close 
to the information-theoretic lower bound, while supporting
 efficient navigational operations.
This technique was successfully applied to various other abstract
data types, such as dictionaries, strings, binary
relations~\cite{bgmr2006, bhmr2007}
and labeled
trees~\cite{grr2004,bgmr2006, bhmr2007}. 
It has also been applied to data structures related to computational geometry. 
There are succinct representations of planar triangulations and planar graphs~\cite{cll2001, cds2005, cds2006, bchm2007} that use $O(n)$ bits, and support queries such as testing the adjacency between two given vertices in constant time.    
However, they only encode the connectivity information, so they are succinct graph data structures rather than succinct geometric data structures. 
It is not known how to combine them with point location structures without using $O(n)$ extra words or $O(n\lg n)$ bits. 

In this paper, we propose to design succinct geometric data structures. 
Given a geometric data set, our goal is to store the coordinates of the points as a permuted sequence, and design an auxiliary data structure called {\em succinct geometric index} that occupies negligible space (more precisely, $o(n)$ bits) to support various geometric queries in optimal time. 
There is similarity between this model and the permutation + bits model presented by Chan and Chen~\cite{CC08}, but they are different. 
The latter was proposed only as an intermediate model for the design of implicit geometric data structures, in which $O(n)$ bits are allowed in addition to storing a permutation of the points. 
It also only allows bit probe operations to these bits, while we do not have such a restriction. 

\subsection{Our Results}
We design succinct geometric indexes to answer point location queries. 
Our first and main result is that, given a planar triangulation, we can permute its points to store the coordinates in sequence, and construct a succinct geometric index that occupies $o(n)$ bits to support point location in $O(\lg n)$ time. 
The preprocessing time is $O(n)$. 
Based on this, we design three variants of this index. 
The first variant is a succinct geometric index that supports point location in $\lg n + 2\sqrt{\lg n} + O(\lg^{1/4} n)$ steps, which matches the result of Seidel and Adamy~\cite{sa2000} while using negligible space. 
The preprocessing is $O(n)$, which is an improvement upon the $O(n\lg n)$ preprocessing time of the latter structure. 
The second variant is a succinct geometric index that supports point location in $o(\lg n)$ time when the coordinates are integers bounded by $U$. 
The last variant is a  succinct geometric index that can answer point location in $O(H+1)$ expected time. 
These results match the query efficiency of previous point location structures that use $O(n)$ words or $O(n\lg n)$ bits, while saving drastic amounts of space.

We then generalize our approach to the case of planar subdivisions (we assume the subdivision is within a polygon boundary), and design an $o(n)$-bit index that supports point location on planar subdivisions in $O(\lg n)$ time. 
This immediately yields another succinct geometric index that can test whether a query point is inside a given polygon. 
We also use our techniques for point location to design a succinct geometric index that supports vertical ray shooting. 
Finally, we apply our succinct geometric indexes to design the first implicit data structures that support point location in $O(\lg^2 n)$ time. 
All our results are under the word RAM model of $\Theta(\lg n)$-bit word size.

\section{Preliminaries}

\subsection{Bit Vectors}
\label{sec:bitvector}

A key structure for many succinct data structures, and for the research work in this paper, is a bit vector\index{bit vector|(} $B$ of length $n$ that supports
$\rankop$ and $\selop$ operations. We assume that the positions
in $B$ are numbered $1, 2,\ldots, n$.  
For $\lab\in\{0,1\}$, the operator $\rankop_B(\lab,\object)$ returns
the number of occurrences of $\lab$ in $B[1\ldots\object]$, and the
operator $\selop_B(\lab,\rank)$ returns the position of the
$\rank$\textsuperscript{th} occurrence of $\lab$ in $B$.
We omit the subscript $B$ when it is clear from the context.  
Lemma~\ref{lem:ranksel} addresses the problem of succinctly representing bit vectors, in which part (a) is
from Jacobson~\cite{j1989} and Clark and
Munro~\cite{cm1996}, while part (b) is
from Raman~\etal~\cite{rrr2002}.

\begin{lemma} \label{lem:ranksel}
  A bit vector $B$ of length $n$ with $v$ $1$s can be represented
  using either: (a) $n + o(n)$ bits, or (b) $\lg{n \choose v} +
  O(n\lg\lg n / \lg n)$ bits, to support the access to each bit,
  $\rankop$ and $\selop$ in $O(1)$ time.
\end{lemma}

\subsection{Graph Separators}
Graph separators~\cite{lt1980} have been extensively studied to partition graph into subgraphs, to allow divide and conquer. 
The variant of graph separators we use is called {\em $t$-separators}. 
Let $G=\{V, E\}$ be a planar graphs of $n$ vertices, where each vertex has a non-negative weight. 
A {\em $t$-separator} ($0 < t < 1$) of $G$ is a subset, $S$, of $V$ whose removal from $G$ leaves no connected component of total weight more than $w(G)$, where $w(G)$ is the sum of the weights of the vertices of $G$. 
Aleksandrov and Djidjev~\etal~\cite{ad1996} have the following results:

\begin{lemma}[\cite{ad1996}]
\label{lem:tseparator}
\label{lem:tseparatorweight}
Consider a planar graph $G$ with $n$ vertices, whose vertices have nonnegative weights. 
For any $t$ such that $0<t<1$, there is a $t$-separator consisting of $O(\sqrt{n/t})$ vertices that can be computed in $O(n)$ time. 
\end{lemma}

\begin{REMOVED}
Graph separators~\cite{lt1980} have been extensively studied to partition graph into subgraphs, to allow divide and conquer. 
The variant of graph separators we use is called {\em $t$-separators}. 
Let $G=\{V, E\}$ be a planar graphs of $n$ vertices. 
A {\em $t$-separator} ($0 < t < 1$) of $G$ is a subset, $S$, of $V$ whose removal from $G$ leaves no connected component of more than $tn$ vertices. 
Aleksandrov and Djidjev~\etal~\cite{ad1996} considered the problem of partitioning graphs with weights and costs. 
Their results can be directly applied to unweighted planar graphs: 

\begin{lemma}[\cite{adgm2006}]
\label{lem:tseparator}
Consider a planar graph $G$ with $n$ vertices. 
For any $t$ such that $0<t<1$, there is a $t$-separator consisting of $O(\sqrt{n/t})$ vertices that can be computed in $O(n)$ time. 
\end{lemma}

When weights are assigned to vertices, a $t$-separator is defined as a set of vertices whose removal from $G$ leaves no connected component of total weight more than $w(G)$, where $w(G)$ is the sum of the weights of the vertices of $G$. 
The following lemma considers this case: 

\begin{lemma}[\cite{adgm2006}]
\label{lem:tseparatorweight}
Consider a planar graph $G$ with $n$ vertices, whose vertices have nonnegative weights. 
For any $t$ such that $0<t<1$, there is a $t$-separator consisting of $O(\sqrt{n/t})$ vertices that can be computed in $O(n)$ time. 
\end{lemma}
\end{REMOVED}

\subsection{Encoding a Planar Triangulation by Permuting its Vertex Set}
Denny and Sohler~\cite{ds1997} considered the problem of encoding the connectivity information of a planar triangulation by permuting its vertex set. They have the following result: 
\begin{lemma}[\cite{ds1997}]
\label{lem:encodetriag}
Given a planar triangulation of $n$ vertices where $n > 1090$, there is an algorithm that can encode it as a permutation of its point set in $O(n)$ time, such that it can be decoded from this permutation in $O(n)$ time. 
\end{lemma}

\section{Point Location in Planar Triangulations}
\label{sec:triangulated}

In this section, we show how to design succinct geometric indexes to support point location queries on a planar triangulation $G$ of $n$ vertices, $m$ edges and $f$ internal faces. 
We define a planar triangulation to be a planar subdivision in which each face (including the outer face) is a triangle. 
For simplicity, we use the term planar triangulation to refer to both the triangulation itself (with coordinates), and the embedded abstract planar graph underlying it. 
We start with our scheme of partitioning the triangulation. 
We then show how to label its vertices (i.e. how to permute its point set). 
We finally design data structures and algorithms to support point location queries.

\subsection{Partitioning a Planar Triangulation by Removing Faces}
\label{sec:partriag}

In this section, we present an approach to partition a planar triangulation by removing a set of internal faces. 

We first define the following terms. 
In a planar triangulation, two faces of the graph are {\em adjacent} if they have a common edge. 
A {\em face path} is a sequence of the faces of the graph such that each two consecutive faces in this sequence are adjacent. 
An {\em adjacent face component} is a set of internal faces of the graph in which there exists a face path between any two faces in this set, and there is no face in this set that is adjacent to an internal face not in this set. 
We define the {\em size} of an adjacent face component to be the number of internal faces in it. 
With these we can define the notion of graph separators consisting of faces: 

\begin{definition}
Consider a planar triangulation with $f$ internal faces. 
A {\bf $\boldmath{t}$-face separator} ($0 < t < 1$) of $G$ is a set of its internal faces of size $O(\sqrt{f/t})$ whose removal from $G$ leaves no adjacent face component of more than $tf$ faces. 
\end{definition}

We have the following lemma: 

\begin{lemma}
\label{lem:tface}
Consider a planar triangulation $G$ of $f$ internal faces. 
For any $t$ such that $0<t<1$, there is a $t$-face separator consisting of $O(\sqrt{f/t})$ faces that can be computed in $O(n)$ time. 
\end{lemma}

\begin{proof}
We consider graph $G^*$, which is the dual graph of $G$ excluding the vertex corresponding to the outer face of $G$ and its incident edges. 
Then $G^*$ has $f$ vertices. 
As $G$ is a planar triangulation, $G^*$ is a simple planar graph. 
By Lemma~\ref{lem:tseparator} (we simply let the weight of each vertex be $1$), there exists a $t$-separator, $S^*$, for $G^*$, consisting of $O(\sqrt{f/t})$ vertices. 
We use $S$ to denote the set of faces of $G$ corresponding to the vertices in $S^*$. 
Then $S$ has $O(\sqrt{f/t})$ faces. 
As the removal of $S^*$ from $G^*$ leaves no connected component of more than $tf$ vertices, the removal of $S$ from $G$ leaves no adjacent face component of more than $tf$ faces. 
Therefore, $S$ is a $t$-face separator of $G$. 
\qed
\end{proof}

\begin{figure}[t]
\centering
\epsfig{file = 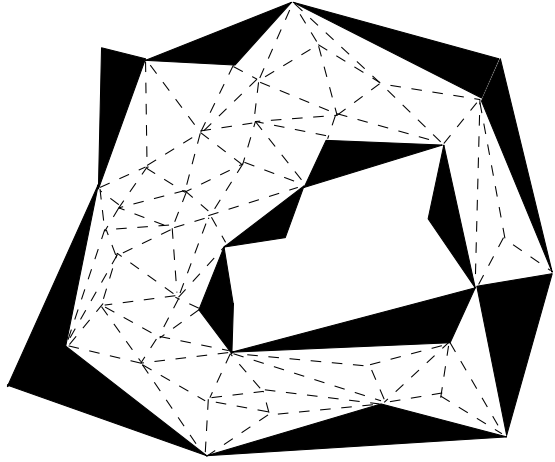}
\caption{A typical adjacent face component, triangulated using dashed lines. The black triangles are separator triangles that are adjacent to the triangles of this component. The other parts of the graph are not shown.}
\label{figure:comp}
\end{figure}

Figure~\ref{figure:comp} shows a sample adjacent face component. 
We define the {\em boundary} of an adjacent face component to be a set of edges in which each edge is shared by an internal face of this component and a face of the $t$-face separator. 
Thus the boundary of an adjacent face component consists of one or more simple cycles: one simple cycle which is the outer face of the adjacent face component, and at most one simple cycle corresponding to each adjacent face component inside it. 
The simple cycle corresponds to the outer face does not share an edge with any cycle inside; otherwise, the cycle inside is simply part of the outer face. 
A useful observation is that any two such simple cycles do not have a common edge, because otherwise, there are two faces of $G$ sharing an edge that are in two different adjacent face components, which contradicts the definition of adjacent face components. 

As each adjacent face component obtained using Lemma~\ref{lem:tface} may have $1$ to $t n$ vertices, there may be as many as $O(n)$ of them. 
To further bound the number of subgraphs we partition the graph into, we have the following lemma:

\begin{lemma}
\label{lem:componentno}
Consider a planar triangulation $G$ with $f$ internal faces and a $t$-face separator $S$ constructed using Lemma~\ref{lem:tface}. 
The number of adjacent face components of $G\setminus S$ is $O(\sqrt{f/t})$. 
\end{lemma}
\begin{proof}
As the number of adjacent face components is always less than the number of edges in their boundaries, and no edge is in the boundaries of two different components, we need only bound the number of such edges. 
Observe that these edges are from the set $S\cup\{\{v_0, v_1\}, \{v_1,v_{n-1}\}, \{v_{n-1}, v_0\}\}$, where $v_0$, $v_1$ and $v_n$ are the three vertices on the outer face of $G$. 
Therefore, the number of such edges is at most linear in the size of $S$, which is $O(\sqrt{f/t})$. \qed
\end{proof}

A {\em boundary vertex} of an adjacent face component is a vertex on the boundary of the adjacent face component, and an {\em internal vertex} is a vertex inside it. 
A vertex {\em belongs to} an adjacent face component iff it is either a boundary vertex or an internal vertex of this component. 
The {\em duplication degree} of each vertex is the number of adjacent face components it belongs to. 
Thus the duplication degree of an internal vertex is $1$. 
To bound the sum of the duplication degrees of boundary vertices, we have the following lemma:

\begin{lemma}
\label{lem:dup}
Consider a planar triangulation $G$ with $f$ internal faces and a $t$-face separator $S$ constructed using Lemma~\ref{lem:tface}. 
The sum of the duplication degrees of all its boundary vertices is $O(\sqrt{f/t})$. 
\end{lemma}
\begin{proof}
Recall that the boundary of an adjacent face component consists of a set of simple cycles. 
Thus the duplication degree of a boundary vertex is the number of simple cycles it is in. 
Therefore, the sum of the duplication degrees of all the boundary vertices is equal to the number of edges in all such simple cycles. 
As each such edge is an edge of a face of $S$, and no edge exists in two different cycles, the number of such edges is at most three times the number of faces of $S$, which is $O(\sqrt{f/t})$. \qed
\end{proof}

\begin{REMOVED}
As each adjacent face component obtained using Lemma~\ref{lem:tface} may have $1$ to $t n$ vertices, there may be as many as $O(n)$ of them. 
To further limit the number of subgraphs we partition the graph into, we introduce the notion of {\em regions}. 
We define a region to be the union of one or more adjacent face components. 
We define the {\em size} of a region to be the number of internal faces in it. 
We have the following lemma. 

\begin{lemma}
\label{lem:tregion}
Consider a planar triangulation $G$ with $n$ vertices and $f$ internal faces and a $t$-face separator, $S$, of $G$ consisting of $O(\sqrt{f/t})$ faces. 
There exists an $O(n)$-time algorithm that can group the adjacent face components obtained by removing $S$ from $G$ into at most $2/t$ regions, such that the size of each region is at least $tf/2$ and less than $3tf/2$. 
\end{lemma}

\begin{proof}
We first visit each adjacent face component once and create a set of regions from them. 
During this process, there may exist one region whose size is less than $tf/2$, and we call it a {\em partial region}. 
We keep track of the number of the internal faces in the partial region when it exists. 

For each adjacent face component, we first count the number of faces in it. 
If its size is at least $tf/2$, we make it a region and visit the next adjacent face component. 
If its size is less than $tf/2$, we check if there is a partial region. 
If there is not, we create a new partial region consisting of this adjacent face component. 
If there is, we add this adjacent face component into the partial region and update the information of the size of the partial region. 
If its size is greater than or equal to $tf/2$, we make it a region. 
At the end of this process, we have a set of of regions whose sizes are at least $tf/2$ and at most $tf$, and possibly one partial region with less than $tf/2$ internal faces. 
We merge the partial region with any region, and the size of the resulting region is at least $tf/2$ and less than $3tf/2$. 

The above process clearly takes $O(n)$ time. \qed
\end{proof}

We define the {\em boundary} of an adjacent face component (or region) to be a set of edges in which each edge is both an edge of one internal face of the adjacent face component (or region) and an edge of a face in the $t$-face separator. 
Thus the boundary of an adjacent face component (or region) consists of one or more simple cycles: one simple cycle which is the outer face of the adjacent face component (or region), and one simple cycle corresponding to each adjacent face component (or region) inside it. 
A {\em boundary vertex} of an adjacent face component (or region) is a vertex on the boundary of the adjacent face component (or region), and an {\em internal vertex} is a vertex inside it. 
The {\em duplication degree} of each vertex is the number of regions it belongs to. 
Thus the duplication degree of an internal vertex is $1$. 
To bound the sum of the duplication degrees of boundary vertices of the regions, we have the following lemma:

\begin{lemma}
Consider a planar triangulation $G$ with $f$ internal faces and a $t$-face separator, $S$, of $G$ consisting of $O(\sqrt{f/t})$ faces, and partition it into regions using the approach above. 
The sum of the duplication degrees of all the boundary vertices of all its regions is $O(\sqrt{f/t})$. 
\end{lemma}
\begin{proof}
Recall that the boundary of a region consists of a set of simple cycles as stated above. 
Thus the duplication degree of a boundary vertex is the number of simple cycles it is in. 
Therefore, the sum of the duplication degrees of all the boundary vertices is equal to the number of edges in all such simple cycles. 
As each such edge is an edge of a distinct internal triangle in $S$, the number of such edges is at most three times the number of internal triangles in $S$, which is $O(\sqrt{f/t})$. \qed
\end{proof}
\end{REMOVED}

\subsection{The Two-Level Partitioning Scheme}
\label{sec:twolevel}
We now use Lemma~\ref{lem:tface} to partition the input graph $G$. 
Recall that $G$ has $n$ vertices and $f$ internal faces. 
Thus $f = 2n-5$. 

We first use Lemma~\ref{lem:tface} to partition $G$. 
We choose $t=(\lg^a f) / f$, where $a$ is a positive constant parameter that we will fix later. 
Then the $t$-face separator, $S$, has $O(\sqrt{f/t}) = O(f / \lg^{a/2} t)$ faces and thus $O(n / \lg^{a/2} n)$ vertices. 
We call each adjacent face component of $G\setminus S$ a {\em region}. 
By Lemma~\ref{lem:componentno}, there are  $r=O(\sqrt{f/t}) = O(f / \lg^{a/2} t)$ regions. 
Each region has at most $tf = \lg^a f$ internal faces, and thus $O(\lg^a n)$ vertices. 
We use $R_i$ to denote the $i${\kth} region of $G$ (the relative order of regions does not matter). 
We call this the {\em top-level partition} of $G$. 

We perform another level of partitioning. 
For each region $R_i$, we triangulate the graph that consists of $R_i$ and the triangular outer face of $G$, and we denote the resulting planar triangulation $R_i'$. 
Thus $R_i'$ has $O(\lg^a n)$ vertices. 
We partition $R_i'$ to smaller ``regions'' called {\em subregions} so that each subregion has $O(\lg^b n)$ vertices, where $b$ is a positive constant parameter smaller than $a$ that we will fix later. 
We use $R_{i,j}$ to denote the $j${\kth} region of $R_i$ (the relative order of subregions in the same region does not matter). 
To do this, let $n_i$ be the number of vertices in $R_i$. 
If $n_i > \lg^b n$ (otherwise, the entire region is also a subregion and the separator has size $0$), we choose $t_i= (\lg^b n) / n_i$ and use Lemma~\ref{lem:tface} to construct a $t_i$-face separator, $S_i$, for each region $R_i$. 
Then $S_i$ has $O(\sqrt{n_i /(\lg^b n / n_i)}) = O(n_i / \lg^{b/2} n)$ vertices. 
The sum of the numbers of vertices in all the $S_i$'s is $\Sigma_{i=1}^rO(n_i / \lg^{b/2} n) = O(n /\lg^{b/2} n)$. 
By Lemma~\ref{lem:componentno}, there are  $O(\sqrt{n_i/t_i}) = O(n_i / \lg^{b/2} n)$ subregions in $R_i$.
Therefore, the total number of subregions in $G$ is $\Sigma_{i=1}^rO(n_i / \lg^{b/2} n) = O(n /\lg^{b/2} n)$. 
We call this the {\em bottom-level partition} of $G$.

\subsection{The Labeling of the Vertices}
\label{sec:label}

We now design a labeling scheme for the vertices based on the two-level partition in Section~\ref{sec:twolevel}. 
This labeling scheme assigns a distinct number from the set $[n]$ to each vertex $x$ of the graph\footnote{We use $[i]$ to denote the set $\{1, 2, ..., i\}$.}. 
We call this number the {\em graph-label} of $x$. 
For each region $R_i$, this labeling scheme also assigns a distinct number from the set $[n_i]$ to each vertex $x$ in $R_i$, where $n_i$ is the number of vertices in this region. 
We call this number the {\em region-label} of $x$. 
For each subregion $R_{i,j}$, a unique number from the set $[n_{i,j}]$ is assigned to each vertex in $R_{i,j}$, where $n_{i,j}$ is the number of vertices in this subregion. 
We call this number the {\em subregion-label} of $x$.

Observe that, although each vertex $x$ has one and only one graph-label, it may have zero, one or several region-labels or subregion-labels. 
This is because each vertex may belong to more than one region or subregion, or only belong to the $t$-separator of $G$ or a $t_i$-separator of a region $R_i$ of $G$. 

We assign the labels from bottom up. 
We first assign the subregion-labels. 
Given subregion $R_{i,j}$, 
we use Lemma~\ref{lem:encodetriag} to permute its vertices (we have to surround $R_{i,j}$ using a triangle and triangulate the resulting graph to use this lemma, but as the vertices of the added triangle are always the last three vertices when permuted, this does not matter). 
If a vertex $x$ in $R_{i,j}$ is the $k${\kth} vertex in this permutation, then the subregion-label of $x$ in $R_{i,j}$ is $k$. 

To assign a region-label to a vertex $x$ that belongs to region $R_i$, there are two cases. 
First, we consider the case where $x$ belongs to at least one subregion in $R_i$. 
Let $h_i$ be the number of vertices in $R_i$ that belongs to at least one subregion of $R_i$. 
We assign a distinct number from $[h_i]$ to each such vertex as its region-label in the following way:  
We visit each subregion $R_{i, j}$, for $i = 1, 2, \cdots, u_i$, where $u_i$ is the number of subregions in $R_i$. 
When we visit $R_{i,j}$, we check all its vertices sorted by their subregion-labels in increasing order. 
We output a vertex of $R_{i,j}$ iff we have not checked this vertex before (i.e. it does not belong to subregions $R_{i, 1}, R_{i,2}, \cdots, R_{i, j-1}$). 
This way we output each vertex that belongs to at least one subregion in $R_i$ exactly once, and we assign the number $k$ to the $k${\kth} vertex we output, and this number is its region-label in $R_i$. 
Second, we consider the case where $x$ does not belong to any subregion in $R_i$. 
There are $n_i - h_i$ such vertices. 
We assign a distinct number from $\{h_i+1, h_i+2, \cdots, n_i\}$ to each of them in an arbitrary order, and the numbers assigned are their region-labels.

We assign graph-labels to the vertices using an approach similar to the one in the previous paragraph. 
More precisely, we permute the $h$ vertices that have region-labels in the order we first visit them when we check all the vertices by region. 
The $k${\kth} vertex in such a permutation has graph-label $k$. 
We assign a distinct graph-label from $\{h+1, h+2, \cdots, n\}$ to the rest of the vertices of the graph in an arbitrary order.

We now show how to perform constant-time conversions between the graph-labels, region-labels and subregion-labels of the vertices of $G$. We have the following lemma:

\begin{lemma}
\label{lem:labconversion}
There is a data structure of $o(n)$ bits such that given a vertex $x$ as a subregion-label $k$ in subregion $R_{i,j}$, the region-label of $x$ in $R_i$ can be computed in $O(1)$ time. 
Similarly, there is a data structure of $o(n)$ bits such that given a vertex $x$ as a region-label $k$ in region $R_i$, the graph-label of $x$ can be computed in $O(1)$ time if $a > 2$. 
\end{lemma}
\begin{proof}
We first show how to prove the first claim of this lemma. 

Recall that we use $u_i$ to denote the number of subregions in $R_i$ and $h_i$ to denote the number of vertices in $R_i$ that have subregion-labels.  
We denote the number of regions of $G$ by $r$. 
We denote the number of vertices of subregions $R_{i,1}, R_{i,2}, \cdots, R_{i, u_i}$ by $n_{i,1}, n_{i,2}, \cdots, n_{i, {u_i}}$, respectively, where $u_i$ is the number of subregions in $R_i$.  
Let $n_i'$ be $\sum_{j=1}^{u_i} n_{i,j}$. 
As no internal vertex of a subregion occurs in another subregion, in order to bound $n_i'$, we need only consider the boundary vertices of the subregions in $R_i$. 
Then by Lemma~\ref{lem:dup}, we have $n_i' = h_i + O(\sqrt{n_i / t_i}) = h_i + O(n_i/ \lg^{b/2} n) \le n_i + O(n_i/ \lg^{b/2} n)$. 

We consider a conceptual array, $A_i$, of length $n_i'$ for each region $R_i$. 
It is the concatenation of the following conceptual arrays. 
For each subregion $R_{i,j}$, we construct a conceptual array $A_{i,j}$ in which $A_{i,j}[k]$ stores the region-label of the vertex in $R_{i,j}$ whose subregion-label is $k$. 
Then $A_i = A_{i,1}A_{i,2}\cdots A_{i,u_i}$. 
Clearly $A_i$ has the answers to our queries, but we do not store it explicitly. 
Instead, we construct the following data structures for each region $R_i$:
\begin{itemize}
\item A bit vector $B_i[1..n_i']$ which stores the numbers $n_{i,1}, n_{i,2}, \cdots, n_{i, u_i}$ in unary, i.e. $B_i = 10^{n_{i,1}-1}10^{n_{i,2}-1}\cdots10^{n_{i,u_i}-1}$; \footnote{We use $0^l$ to denote a bit sequence of $l$ $0$s.}
\item A bit vector $C_i[1..n_i']$ in which $C_i[k] = 1$ iff the first occurrence of the region-label $A_i[k]$ in $A$ is at position $k$ (let $q_i$ denotes the number of $0$s in $C_i$);
\item An array $D_i[1..q_i]$ in which $D_i[k]$ stores the region-label of the vertex that corresponds to the $k${\kth} $0$ in $C_i$, i.e. $D_i[k] = A_i[\selop_{C_i}(0, k)]$.
\end{itemize}

To analyze the space cost of the above data structures constructed for the entire graph $G$, we first show how to store all the $B_i$s. 
The first step is to concatenate all the $B_i$s and store them as a single sequence $B$, and use part (b) of Lemma~\ref{lem:ranksel} to represent $B$. 
Let $y$ be the length of $B$ and $z$ be the number of $1$s in $B$. 
Then $B$ can be stored in $\lg {y \choose z} + O(y\lg\lg y / \lg y)$ bits. 
As $n_i' \le n_i + O(n_i/ \lg^{b/2} n)$, we have $y = \sum_{i=1}^r n_i' \le n + O(n / \lg^{b/2} n)$. 
By Lemma~\ref{lem:componentno}, we have $u_i = O(n_i / \lg^{b/2} n)$, so $z = \sum_{i=1}^r u_i = O(n / \lg^{b/2} n)$. 
As each subregion has $O(\lg^b n)$ vertices, we also have $z = \sum_{i=1}^r u_i = \Omega (n / \lg^b n)$. 
By applying the inequality $\lg {w \choose u} \le u \lg \frac{ew}{u} + O(1)$~\cite[Section~4.6.4]{he2007}, we have $\lg {y \choose z} = \lg {\sum_{i=1}^r n_i' \choose \sum_{i=1}^r u_i} = O(n\lg\lg n / \lg^{b/2} n)$. 
Thus space cost of $B$ is $O(n\lg\lg n / \lg^{b/2} n) + O(n\lg\lg n / \lg n) = o(n)$ bits. 
The rank/select operations on $B$ can be performed in constant time, thus in order to support the same operations on each $B_i$ in constant time, it suffices to locate the starting position of any $B_i$ in $B$ in constant time. 
This can be done by using another bit vector, $X$, of length $y$ to mark the starting positions of all the $B_i$s in $B$. 
Thus the length of $X$ is at most $n + O(n/ \lg^{b/2} n)$ and it has $r=O(n / \lg^{a/2} n)$ $1$s, which can be stored using part (b) of Lemma~\ref{lem:ranksel} in $o(n)$ bits. 
The same scheme can be used to concatenate and store all the $C_i$s, and by a similar analysis, this occupies $O(n\lg\lg n / \lg^{b/2} n + O(n\lg\lg n / \lg n))=o(n)$ bits. 
The same bit vector $X$ allows us to support rank/select on each $C_i$ in constant time. 
Finally, as $q_i$ is less than the sum of the duplication degrees of all the vertices in $R_i$ under the bottom-level partition, by Lemma~\ref{lem:dup}, we have $q_i = O(\sqrt{n_i / t_i}) = O (n_i / \lg^{b/2} n)$. 
As each element of $D_i$ is within the range $[1..n_i']$, it can be stored using $O(\lg\lg n)$ bits. 
Thus $D_i$ occupies $O(n_i \lg\lg n / \lg^{b/2}n)$ bits, so all the $D_i$s occupy $O(n \lg\lg n / \lg^{b/2} n)$ bits. 
They can be concatenated and stored using the same scheme with $o(n)$ additional bits.

\begin{REMOVED}
To analyze the space cost of the above data structures constructed for the entire graph $G$, we first use part (b) of Lemma~\ref{lem:ranksel} to store $B_i$ in $\lg {n_i' \choose  u_i} + o(n_i')$ bits. 
Thus the overall space cost of all the $B_i$s is $\sum_{i=1}^r (\lg {n_i' \choose  u_i} + o(n_i'))$. 
Recall that $n_i' \le n_i + O(n_i/ \lg^{b/2} n)$ and $\sum_{i=1}^r u_i = O(n / \lg^{b/2} n)$. 
Thus $\sum_{i=1}^r n_i' \le n + O(n / \lg^{b/2} n)$. 
Therefore, the above cost is $\sum_{i=1}^r \lg {n_i' \choose  u_i} + o(n) \le \lg {\sum_{i=1}^r n_i' \choose \sum_{i=1}^r u_i} + o(n)$. 
As each subregion has $O(\lg^b n)$ vertices, we also have $\sum_{i=1}^r u_i = \Omega (n / \lg^b n)$. 
By applying the inequality $\lg {w \choose u} \le u \lg \frac{ew}{u} + O(1)$~\cite[Section~4.6.4]{he2007}, we have $\lg {\sum_{i=1}^r n_i' \choose \sum_{i=1}^r u_i} = O(n\lg\lg n / \lg^{b/2} n)$. 
Hence all the $B_i$s occupy $o(n)$ bits. 
We also use Lemma~\ref{lem:ranksel} to store $C_i$, and by a similar analysis, all the $C_i$s occupy $O(n\lg\lg n / \lg^{b/2} n)=o(n)$ bits. 
Finally, as $q_i$ is less than the sum of the duplication degrees of all the vertices in $R_i$ under the bottom-level partition, by Lemma~\ref{lem:dup}, we have $q_i = O(\sqrt{n_i / t_i}) = O (n_i / \lg^{b/2} n)$. 
As each element of $D_i$ is within the range $[1..n_i']$, it can be stored using $O(\lg\lg n)$ bits. 
Thus $D_i$ occupies $O(n_i \lg\lg n / \lg^{b/2}n)$ bits, so all the $D_i$s occupy $O(n \lg\lg n / \lg^{b/2} n)$ bits. 
Therefore, all the above data structures occupy $O(n\lg\lg n / \lg^{b/2} n)=o(n)$ bits. 

For each region $R_i$, $\Theta(\lg n)$ bits are enough to store pointers to the corresponding $B_i$, $C_i$ and $D_i$ constructed above, as well as their lengths, so that we can access each element of these data structures in constant time. 
This takes $O(n/\lg^{a/2-1}n)$ bits, which requires $a > 2$ so that the space cost is $o(n)$ bits. 
A better scheme can be used to remove this restriction, which is to concatenate all the $B_i$s and store them as a single sequence $B$, and use another bit vector, $B'$, of length $|B|$ to mark the starting positions of all the $B_i$s in $B$. 
The length $B'$ is at most $n_i + O(n_i/ \lg^{b/2} n)$ and it has $r=O(n / \lg^{a/2} n)$ $1$s, which can be stored using part (b) of Lemma~\ref{lem:ranksel} in $o(n)$ bits without requiring $a$ to be greater than $2$. 
The rank/select operations on $B'$ are sufficient to locate the starting position of any $B_i$ in $B$ and compute its size in constant time. 
The same scheme can be used to store all the $C_i$s and $D_i$s so that the additional cost is $o(n)$ bits for any positive number $a$. 
\end{REMOVED}

We now show how to compute, given a vertex $x$ with subregion-label $k$ in subregion $R_{i,j}$, the region-label of $x$ in $R_i$. 
We first locate the position, $l$, in $A_i$ that corresponds to the occurrence of vertex $x$ in subregion $R_{i,j}$. 
As the vertex with subregion-label $1$ in $R_{i,j}$ corresponds to position $\selop_{B_i}(1, j)$ in $A_i$, we have $l = \selop_{B_i}(1, j) + k - 1$. 
We then retrieve $C_i[l]$. 
If it is $1$, then the first occurrence of the region-label of $x$ in $A$ is at position $l$, thus $\rankop_{C_i}(1, l)$ is the result by the labeling scheme. 
Otherwise, the result is explicitly stored in $D_i[\rankop_{c_i}(0, l)]$. 
The above operations clearly take constant time. 

The second claim of the lemma can be proved similarly, and the space required is $O(n / \lg^{a/2-1} n) + o(n)$ bits, which is $o(n)$ bits when $a > 2$. 
\qed
\end{proof}

\subsection{Answering Point Location Queries}

\begin{lemma}
\label{lem:pointlocation}
Given a planar triangulation $G$ of $n$ vertices, 
there is a succinct geometric index of $o(n)$ bits that supports point location on $G$ in $O(\lg n)$ time. 
\end{lemma}
\begin{proof}
We perform the two-level partitioning of $G$ as in Section~\ref{sec:twolevel}, and use the approach in Section~\ref{sec:label} to assign labels to the vertices of $G$, but we do not store these labels. 
Instead, we sort the vertices by their graph-labels in increasing order, and store their coordinates as a sequence. 
We then show how to construct a succinct geometric index of $o(n)$ bits. 

The succinct geometric index consists of three sets of data structures. 
This first set of data structures are the data structures constructed in Lemma~\ref{lem:labconversion} that supports conversions between subregion-labels, region-labels and graph-labels. 
By the proof of Lemma~\ref{lem:labconversion}, they occupy $O(n/\lg^{a/2-1} n) + o(n)$ bits. 
The second and the third sets of data structures correspond to the top-level and the bottom-level partitions. 

To design the data structures for the top-level partition, we consider the graph $S'$ constructed by triangulating the graph consisting of the separator $S$ and the outer face of $G$. 
$S'$ is a planar triangulation of $O(n / \lg^a n)$ vertices, so we can use the approach of Kirkpatrick~\cite{ki1983} (in fact, any structure that uses $O(n\lg n)$ bits for an $n$-vertex planar triangulation to answer point location in $O(\lg n)$ time can be used here) to construct a data structure $P$ of $O(n/\lg^a n)$ words (i.e. $O(n/ \lg^{a-1} n)$ bits) to support the point location queries on $S'$. 
Note that when we construct $P$, we simply use the graph-label of any vertex to refer to its coordinates, so that we do not store any coordinate in $P$. 
For each face of $S'$, we store an integer. If this face is in region $R_i$, we store $i$. 
We store $0$ if it is a face in separator $S$. 
As there are $O(n/\lg^a n)$ faces and the number assigned to each face can be stored in $\lg n$ bits, the space required to store these numbers is $O(n/ \lg^{a-1} n)$ bits. 
All these ($P$ and the numbers assigned to the faces of $S'$) are the data structures for the top-level partition of $G$, and they occupy $O(n/ \lg^{a-1} n)$ bits. 

The data structures for the bottom-level partition are constructed over the regions of $G$. 
Given a region $R_i$, recall that we construct a planar triangulation $R_i'$ when we perform the bottom-level partition. 
Consider the graph $S_i'$ constructed by triangulating the graph consisting of the separator $S_i$ (recall that it is a separator for $R_i'$) and the outer face of $R_i'$. 
We observe that $S_i'$ is a planar triangulation of $O(r_i / \lg^{b/2} n)$ vertices. 
As the number of vertices of $S_i'$ is $O(\lg^a n)$, a pointer that refers to a vertex of $S_i'$ can be stored in $O(\lg\lg n)$ bits. 
To refer to the coordinates of any vertex of $S_i'$, we only uses its region-label as we can compute its graph-label in constant time by Lemma~\ref{lem:labconversion}, and $O(\lg\lg n)$ bits are sufficient to store a region-label. 
Thus we can use the approach of Kirkpatrick~\cite{ki1983} to construct a data structure $P_i$ of $O(r_i\lg\lg n/\lg^{b/2} n)$ bits to support the point location queries on $S_i'$ in $O(\lg\lg n)$ time. 
We also store a number for each face of $S_i'$, and this number is $j$ if this face is in subregion $R_{i,j}$, and $0$ if it is a face in separator $S_i$. 
As there are $O(r_i / \lg^{b/2} n) = O(\lg^{a-b/2} n)$ subregions in $R_i'$, each number occupies $O(\lg\lg n)$ bits, so the space cost of storing these numbers is $O(r_i\lg\lg n/\lg^{b/2} n)$ bits. 
The space cost of these data structures for all the regions is $\sum_{i=1}^r O(r_i\lg\lg n/\lg^{b/2} n) = O(n\lg\lg n / \lg^{b/2} n) = o(n)$ bits. 

Therefore, the succinct geometric index constructed above occupies $O(n/\lg^{a/2-1} n)+o(n)$ bits. 

We now show how to answer point location queries using this succinct geometric index. 
Given a query point $x$, we first locate the face of $S'$ that contains $x$ using the set of data structures constructed for the top-level partition in $O(\lg n)$ time. 
We retrieve the integer, $i$, assigned to the face of $S'$ that $x$ is in. 
If $i$ is $0$, then this face is in $S$, and we return its three vertices as the result. 
If it is not, then $x$ is inside region $R_i$. 
We then use the bottom-level partition to perform a point location query on the graph $S_i'$ in $O(\lg\lg n)$ time using $x$ as the query point. 
We retrieve the integer, $j$, assigned to the face of $S_i'$ that $x$ is in. 
If $j$ is $0$, then this face is in $S_i$, and we return its three vertices as the result (we need convert the region-labels of these three vertices to their graph-labels when we return them). 
If $j$ is not, then $x$ is inside region $R_{i,j}$. 
Using the bit vector $B_i$ constructed in the proof of Lemma~\ref{lem:labconversion}, we can compute the number of vertices of $R_{i,j}$ in constant time. 
Recall that $r_{i,j}$ denotes this number. 
By Lemma~\ref{lem:labconversion}, we can compute the graph-label of any of these vertices in constant time, and thus retrieve its coordinates in $O(1)$ time. 
As we number these vertices using Lemma~\ref{lem:encodetriag}, we can use Lemma~\ref{lem:encodetriag} to construct the graph $R_{i,j}'$ in $O(r_{i,j})$ time, and then check each of its faces to find out the face that $x$ is in. 
This takes $O(r_{i,j})=O(\lg^b n)$ time. 
Therefore, the entire process takes $O(\lg n + \lg^b n)$ time. 

We have so far designed a succinct geometric index of $O(n/\lg^{a/2-1} n)+o(n)$ bits to support point location queries in $O(\lg n + \lg^b n)$ time. 
Choosing $a=3$ and $b=1$ yields a succinct geometric index of $o(n)$ bits that supports point location queries in $O(\lg n)$ time, which completes the proof. \qed
\end{proof}

\begin{lemma}
\label{lem:plpre}
The data structures of Lemma~\ref{lem:pointlocation} can be constructed in $O(n)$ time. 
\end{lemma}
\begin{proof}
We first show how to compute the order of the vertices in $O(n)$ time. 
To prove this, we need show that the two-level partition in Section~\ref{sec:twolevel} and the labeling of the vertices in Section~\ref{sec:label} can be performed in $O(n)$ time. 

To show that the two-level partition can be performed in $O(n)$ time, we first observe that the computation of $t$-face separators at both levels can be performed in time linear in the numbers of the vertices of the graphs. 
Thus such computation can be performed in $O(n)$ time. 
The only part that is not clear is the time required to construct the graphs $R_i'$. 
Recall that we construct $R_i'$ by triangulating the graph that consists of $R_i$ and the outer face of $G$, and that the boundary of $R_i$ consists of one or more simple cycles: one simple cycle which is the outer face of $R_i$, and at most one simple cycle corresponding to each adjacent face component inside it. 
Thus we need only to triangulate the interior of each simple cycle corresponding to each adjacent face component inside $R_i$, and the face, $F$, defined by the simple cycle which is the outer face of $R_i$ and the boundary of the outer face of $G$. 
To triangulate the interior of each simple cycle, we use the linear-time algorithm by Chazelle~\cite{ch1991}. 
To triangulate $F$, 
we start at an arbitrary vertex $y$ on the outer face of  $R_i$. 
We locate a vertex of the triangular outer face of $G$ such that the line segment between this vertex and $y$ does not cross any edge of the outer face of $R_i$, and we draw an edge between $x$ and this vertex. 
This divides $F$ into two simple polygons, which can be triangulated in linear time. 
\begin{REMOVED}
we first use the linear-time algorithm by McCallum and Avis~\cite{ma1979} to compute the convex hull of the outer face of $R_i$. 
Each face created between the convex hull of $R_i$ and the outer face of $R_i$ is a simple polygon and we use the linear-time algorithm by Chazelle~\cite{ch1991} to triangulate it. 
Now we need only triangulate the face $F'$ which is defined by the convex hull of $R_i$ and the boundary of the outer face of $G$. 
To do this, we start with an arbitrary vertex $x$ that is on the convex hull of $R_i$. 
We locate a vertex on the outer face of $G$ such that the line segment between this vertex and $x$ does not cross any edge of the convex hull of $R_i$, and we draw an edge between $x$ and $v_0$. 
Assume, without loss of generality, that this vertex is $v_0$ (we also assume that, when listed in counterclockwise order, the three vertices on the outer face of $G$ are $v_0$, $v_1$ and $v_{n-1}$). 
We then visit the vertex, $y$, that is next to $x$ on the convex hull of $R_i$ in counterclockwise order. 
If vertex $x$ is to the right of the line segment between $v_0$ and $y$, then the triangle defined by $x$, $y$, $v_0$ is outside the convex hull of $R_i$. 
We draw an edge between $v$ and $y$, and this triangle is a face of $R_i'$. 
If not, then the line segment between $y$ and $v$ crosses at least one edge of the convex hull of $R_i$. 
We then find the next vertex after $v_0$ on the outer face of $G$ in counterclockwise order such that the line segment between this vertex and $y$ does not cross any edge of the convex hull of $R_i$, and repeat the above process. 
As there are only three vertices on the outer face of $G$, this process takes time linear in the number of vertices on the convex hull of $R_i$. 
\end{REMOVED}
Thus we can construct the graph $R_i'$ in time linear to the number of its vertices after we perform the top-level partition. 
Therefore, it takes $O(n)$ time to construct all the $R_i'$s. 

The linear time construction of the succinct geometric index in Lemma~\ref{lem:pointlocation} directly follows the linear time construct of Kirkpatrick's point location structure~\cite{ki1983} and the data structure for part (b) of Lemma~\ref{lem:ranksel}. 
\qed
\end{proof}

Combining Lemmas~\ref{lem:pointlocation} and \ref{lem:plpre}, we have our first and main result:

\begin{theorem}
\label{thm:pointlocation}
Given a planar triangulation $G$ of $n$ vertices, 
there is a succinct geometric index of $o(n)$ bits that supports point location on $G$ in $O(\lg n)$ time. 
This index can be constructed in $O(n)$ time. 
\end{theorem}

We now design three variants of this succinct geometric index to address the query efficiency with different assumptions. 
We first consider the exact number of point-line comparisons required to answer a query. 
\begin{corollary}
\label{cor:exact}
Given a planar triangulation $G$ of $n$ vertices, 
there is a succinct geometric index of $o(n)$ bits that supports point location on $G$ using at most $\lg n + 2\sqrt{\lg n} + O(\lg^{1/4} n)$ steps. 
This index can be constructed in $O(n)$ time. 
\end{corollary}
\begin{proof}
We use the same approach as that for Theorem~\ref{thm:pointlocation}, but we choose $b = 1/4$. 
When we construct the data structures for the top-level partition, we use the approach by Seidel and Adamy~\cite{sa2000} to construct the point location structure $P$. 
This way point location on $S'$ can be compute in at most $\lg n + 2\sqrt{\lg n} + O(\lg^{1/4} n)$ steps. 
As the point location on $S_i'$ and $R_{i,j}'$ can be supported in $O(\lg\lg n)$ and $O(\lg^b n) = O(\lg^{1/4} n)$ steps, respectively, the overall steps required to answer point location queries on $G$ is at most $\lg n + 2\sqrt{\lg n} + O(\lg^{1/4} n)$. 

The other claims in this corollary are easy to prove. 
Note that when analyzing the preprocessing time, the super-linear preprocessing time of the approach by Seidel and Adamy~\cite{sa2000} is not a problem, 
as we apply it to the graph $S_i'$, which has $O(n/\lg^3 n)$ vertices. 
Thus $P$ can be constructed in $O(n/\lg^3 n \times \lg(n/\lg^3 n)) = O(n/\lg^2 n)$ time. 
\qed
\end{proof}

Note that Corollary~\ref{cor:exact} not only matches the best result~\cite{sa2000} in terms of the exact number of point-line comparisons using negligible space, but also improves the preprocessing time from $O(n\lg n)$ to $O(n)$. 

If all the coordinates are integers bounded by $U \le
2^w$, we have the following variant:
\begin{corollary}
\label{cor:integer}
Assume that all the point coordinates in the plane are integers bounded by $U\le 2^w$. 
Given a planar triangulation $G$ of $n$ vertices, 
there is a succinct geometric index of $o(n)$ bits that supports point location on $G$ in $O(\min\{\lg n/\lg\lg n, \sqrt{\lg U}\}
+ \lg^{\epsilon} n)$ time, for any constant $\epsilon > 0$. 
This index can be constructed in $O(n)$ time. 
%
\end{corollary}

\begin{proof}
We use the same approach as that for Theorem~\ref{thm:pointlocation}, but we choose $b = \epsilon$. 
When we construct the data structures for the top-level partition, we use the approach by Chan~\cite{ch2006} and P\v{a}tra\c{s}cu~\cite{pa2006} to construct the point location structure $P$. 
This way point location on $S'$ can be compute in $O(\min\{\lg n/\lg\lg n,
\sqrt{\lg U}\})$ time. 
As the point location on $S_i'$ and $R_{i,j}'$ can be supported in $O(\lg\lg n)$ time and $O(\lg^b n) = O(\lg^{\epsilon} n)$ time, respectively, the overall time required to answer point location queries on $G$ is $O(\min\{\lg n/\lg\lg n, \sqrt{\lg U}\} +
\lg^{\epsilon} n)$. 


These data structures can still be constructed in $O(n)$ time, as the structure by Chan~\cite{ch2006} and P\v{a}tra\c{s}cu~\cite{pa2006} can be constructed in linear time. \qed
\end{proof}

We then consider the case where query distribution is known. 
\begin{corollary}
Given a planar triangulation $G$ of $n$ vertices, 
there is a succinct geometric index of $o(n)$ bits that supports point location on $G$ in $O(H+1)$ expected time. 
This index can be constructed in $O(n)$ time. 
\end{corollary}
\begin{proof}
If the probability of a face or a set of faces containing a query point is $p$, we say that this face or this set of faces has probability $p$. 
In this proof, we define the $t$-face separator to be the set of faces of $G$ whose removal partitions $G$ into adjacent face components each of which has probability at most $t$. 
We consider graph $G^*$, which is the dual graph of $G$ excluding the vertex corresponding to the outer face of $G$ and its incident edges, and assign the probability of each face of $G$ to its corresponding vertex in $G^*$. 
By Lemma~\ref{lem:tseparatorweight}, the following lemma is immediate:

\begin{lemma}
\label{lem:tfaceprob}
Consider a planar triangulation $G$ of $f$ internal faces, with probability associated to each face. 
For any $t$ such that $0<t<1$, there is a $t$-face separator consisting of $O(\sqrt{f/t})$ faces that can be computed in $O(n)$ time. 
\end{lemma}

Observe that Lemmas~\ref{lem:componentno} and \ref{lem:dup} also apply to $t$-face separators for graphs whose faces are associated with probabilities. 
This is because we prove these two lemmas by bounding the number of edges in the separator, which has nothing to do with probabilities. 

We choose $t=\lg^3 f / f$ to apply Lemma~\ref{lem:tfaceprob} to $G$. 
Let $S''$ be the $t$-face separator. 
Then $S''$ has $O(n/\lg^{3/2} n)$ vertices. 
We call each adjacent face component of $G\setminus S$ a {\em super region}. 
Thus the number of super regions $O(n/\lg^{3/2} n)$ vertices, and the sum of the duplication degrees of the boundary vertices of all the super regions is also $O(n/\lg^{3/2} n)$. 
Note that we can no longer prove that each super region has $o(n)$ vertices; this is because a super region can have a large number of faces with very low probabilities. 
Thus we further perform a two-level partition on each super region as in the proof of Lemma~\ref{lem:pointlocation}. 
Therefore, we actually perform a three-level partition on $G$, and we call them first-level, second-level and third-level partitions from top down. 
It is straightforward to extend the techniques in Section~\ref{sec:twolevel} to this case to compute the permuted sequence of the vertices, and to perform conversions between the labels assigned to the same vertices at different levels of the partition. 

For the first level partition, we construct a triangulated graph $G''$ by triangulating the graph consisting of $S''$ and the outer face of $G$. 
To assign a probability to each face of $G''$, initially we let the probability of each face of $S''$ to be the same as its probability in $G$, and let the probability of any other internal face to be $1/n$. 
However, the sum of all the probabilities of the faces of $G''$ can be larger than $1$, though it is at most $2$. 
We thus reduce the probability of each face of $G''$ by a constant ratio, so that the sum becomes $1$. 
It is clear that the above process reduces the probability of each face by at most half. 
Therefore, the probability of each face of $S''$ in $G''$ is at least half of that in $G$, and the probability of each internal face of $G''$ that is not in $S''$ is at least $1/(2n)$. 
We construct a point location structure, $P''$, for $G''$ using the approach of Iacono~\cite{ia2004}, or any linear-space structure that answers point location in $O(\lg (1/p))$ time, if the query point is contained in a face of probability $p$. 
$P''$ occupies $O(n/\lg^{1/2} n)$ bits. 
We also store additional information for each face of $G''$ to indicate whether it is a face in $S''$ and, if not, which super region it is in. 
For the second-level and third-level structures, we construct data structures similar to those constructed in Lemma~\ref{lem:pointlocation}. 
The algorithm to answer point location queries is similar, except that we now perform operations at three levels of partition. 

To analyze the query time, it is sufficient to show that, if the face, $z$, of $G$ that contains the query point $x$ has probability $p$, the query can be answered in deterministic time $O(min\{\lg n, \lg (1/p)\})$. 
There are two cases. 
First, $z$ is a face in $S''$. 
In this case, we need only use $P''$ to retrieve the result. 
By Iacono's result~\cite{ia2004}, the time required is $O(\min\{\lg n, \lg (1/p')\}$, where $p'$ is the probability of $z$ in $G''$. 
By the analysis in the above paragraph, $p' > p/2$. 
Thus the claim is true in this case. 
Second, $z$ is not a face in $S''$. 
In this case, the query is answered in $O(\lg n)$ time. 
Thus it suffices to prove that $O(\min\{\lg n, \lg (1/p)\}) = O(\lg n)$. 
Recall that each super region has probability at most $t=\lg^3 n / n$. 
As $z$ is part of a super region, we have $p \le \lg^3 n / n$. 
Thus $\lg (1/p) \ge \lg n - 3\lg\lg n$, and the claim follows. 

It is straightforward to show that the space cost is $o(n)$ bits and that the preprocessing time is $O(n)$. 
\qed
\end{proof}

\begin{REMOVED}
\subsection{Practical Considerations}
The previous results we use to prove Theorem~\ref{thm:pointlocation} have practical implementations. 
For example, we can use practical algorithms to triangulate a polygon when implementing our data structures. 
The main challenge of implementing our data structures is the fact that Lemma~\ref{lem:encodetriag} require the number of vertices to be larger than $1090$. 
Recall that we use Lemma~\ref{lem:encodetriag} to encode a subregion of the graph, which has $O(\lg^b n)$ vertices. 
For the current applications, $O(\lg^b n)$ is not sufficiently large. 
However, we can combine a constant number of subregions when we use Lemma~\ref{lem:encodetriag} to encode them. 
This is because the main idea of Lemma~\ref{lem:encodetriag} is to use the permutation of a constant number of vertices to encode enough information, so the number of permuted vertices has to be large enough. 
When we combine a certain constant number of subregions, we can guarantee that this number is large enough. 
We thus expect it feasible to implement our succinct geometric index. 
\end{REMOVED}

\section{Point Location in Planar Subdivisions}
We now generalize the techniques in Section~\ref{sec:triangulated} to design succinct geometric indexes supporting point location queries in general planar subdivisions. 
In this section, we adopt the assumption that a planar subdivision is contained inside a bounding simple polygon (i.e. it does not have any infinite faces), and each face is also a simple polygon.

\subsection{Encoding a Planar Subdivision by Permuting Its Point Set}

We now generalize Lemma~\ref{lem:encodetriag} to the case of planar subdivisions.

\begin{lemma}
\label{lem:encodesub}
Given a planar subdivision of $n$ vertices for sufficiently large $n$, there exists an algorithm that can encode it as a permutation of its point set in $O(n)$ time, such that the subdivision can be decoded from this permutation in $O(n)$ time. 
\end{lemma}
\begin{proof}
To encode the given planar subdivision $G$, we first surround it using a bounding triangle, triangulate it, and denote the resulted graph by $T$. 
This process takes $O(n)$ time using the approach in the proof of Lemma~\ref{lem:plpre}. 
Lemma~\ref{lem:encodetriag} is sufficient to encode $T$, but we need encode more information in order to decode $G$ as we add more edges when triangulating it. 
We show how to encode and decode such information in the rest of this proof. 

To use Lemma~\ref{lem:encodetriag} to encode $T$, we first compute a maximal independent set, $I$, of vertices of $T$ with degree at most $6$. 
The size of $I$ is at least $n/10$ as shown by Denny and Sohler~\cite{ds1997}. 
We remove the vertices in $I$ from $T$ and re-triangulate $T$. 
We then visit the new triangulation in a canonical way such as BFS, and order the vertices in $I$ by the order of the triangles that contain them. 
We divide the vertices in $I$ into sets of the same constant size, and permute each set to encode enough information so that given the new triangulation and a vertex in $I$ (it is sufficient to visit all the faces of the new triangulation to determine each face that contains a vertex in $I$), we know how to insert it into the new triangulation to reconstruct $T$.  
As each vertex in $I$ has degree at most $6$, its removal from $T$ creates a polygon of size at most $6$. 
Based on this, Denny and Sohler~\cite{ds1997} proved that there are at most $41$ possibilities of inserting a vertex. 
Thus each set of $I$ only has to be large enough so that the permutations of its subsets is sufficient to encode $\lg 41$ bits for each point in $I$. 
We continue this process until all the vertices except the vertices in the outer triangular face are removed. 

To modify the above process to encode enough information to indicate which edge of $T$ is present in $G$, we observe that, to decode $T$, when we insert a vertex into the current triangulation, we determine its neighbours in the previous version of this triangulation, remove the edges in the polygon defined by these neighbours, and draw an edge between this vertex and each of its neighbours. 
Therefore, we draw at most $6$ edges when we insert a vertex. 
To encode whether each of these edge is an edge in $G\setminus T$, we need $6$ bits. 
Thus we make each subset of $I$ to be large enough to encode $\lg 41 + 6$ bits of information for each vertex in it. 
When we decode $T$, each time we insert a vertex to the triangulation and draw an edge between it and one of its neighbours, we use the encoded information and store a flag for each edge to indicate whether it is an edge in $G\setminus T$. 
Once we decode $T$, we visit all its edges to remove those in $G\setminus T$ to get $G$. 

The processes of encoding and decoding clearly take $O(n)$ time. 
\qed
\end{proof}

\subsection{Partitioning a Planar Subdivision by Removing Faces}
We first generalize the $t$-face separators defined for planar triangulations to planar subdivisions. Note that the definition of adjacent face component in Section~\ref{sec:partriag} can be directly applied to planar subdivisions. 
We have the following definition: 

\begin{definition}
Consider a planar subdivision $G$ with $f$ internal faces. 
A {\bf $\boldmath{t}$-face separator} of $G$ is a set of its internal faces of size $O(\sqrt{f/t})$ whose removal from $G$ leaves no adjacent face component of more than $tf$ faces. 
\end{definition}

Observe that in the proof of Lemma~\ref{lem:tface}, we do not make use of the fact that each face of a planar triangulation is a triangle. 
Thus we have the following lemma:
\begin{lemma}
\label{lem:tfacesub}
Consider a planar subdivision $G$ with $f$ internal faces. 
For any $t$ such that $0<t<1$, there is a $t$-face separator consisting of $O(\sqrt{f/t})$ faces that can be computed in $O(n)$ time. 
\end{lemma}

We also define the notion of boundary, boundary vertex, internal vertex and duplication degree on planar subdivisions as in Section~\ref{sec:partriag}. 
Same as the case of planar triangulations, we can bound the number of adjacent face components and the sum of the duplication degrees of boundary vertices after removing a $t$-face separator from a planar subdivision. 
The only difference is when we count the number of edges in the $t$-face separator, we can no longer use the fact that each face has $3$ edges. 
Instead, we make use of the maximum number of vertices of any internal face of the planar subdivision to bound these two values. 
Therefore, we have the following two lemmas:

\begin{lemma}
\label{lem:componentnosub}
Consider a planar subdivision $G$ with $f$ internal faces and a $t$-face separator $S$ constructed using Lemma~\ref{lem:tfacesub}. 
The number of adjacent face components of $G\setminus S$ is $O(k\sqrt{f/t})$, where $k$ is the maximum number of vertices of any internal face of $G$. 
\end{lemma}

\begin{lemma}
\label{lem:dupsub}
Consider a planar subdivision $G$ with $f$ internal faces and a $t$-face separator $S$ constructed using Lemma~\ref{lem:tfacesub}. 
The sum of the duplication degrees of all its boundary vertices is $O(k\sqrt{f/t})$, where $k$ is the maximum number of vertices of any internal face of $G$. 
\end{lemma}

\subsection{The Two-Level Partitioning Scheme}
\label{sec:twolevelsub}

We now use Lemma~\ref{lem:tfacesub} to partition the planar subdivision $G$. 
Recall that $n$, $m$ and $f$ denote the numbers of vertices, edges and internal faces of $G$, respectively. 
However, we cannot use Lemma~\ref{lem:tfacesub} directly, as $f$ can be as small as $1$ for any $n$. 
Instead, we divide the faces of the planar subdivision that have sufficiently many vertices into smaller faces whose sizes are bounded by non-constant parameters. 
It may seem odd not to simply divide the faces into triangles, but it is crucial to choose a non-constant parameter for our solution. 
We have the following lemma:

\begin{lemma}
\label{lem:facediv}
Consider a simple polygon $P$ of $n$ vertices. 
Given an integer $l$ where $l < n$, there is an $O(n)$-time algorithm that can, by adding edges between the vertices of $P$ that only intersect at the vertices of $P$, 
divide the interior of $P$ into a planar subdivision such that each internal face has at least $l$ vertices (with the exception of at most one internal face) and at most $3l$ vertices. 
\end{lemma}
\begin{proof}
We first triangulate $P$ in $O(n)$ time~\cite{ch1991} and denote the resulted graph by $P'$. Note that we do not add a triangular outer face when triangulating $P$. 
As each internal face of $P'$ is a triangle, each vertex of the dual graph, $P^*$, of $P'$ (without considering the outer face) has degree at most $3$. 
The BFS tree, $T$, of $P^*$ is thus a binary tree. 
It only suffices to prove that we can partition $T$ into subtrees of size at least $l$ (with exception of at most one subtree) and at most $3l$ in $O(n)$ time. 
One way of achieving this is to apply the partition algorithm of Munro, Raman and Storm~\cite{mrs2001}. 
\qed
\end{proof}

With the above lemma, we can now present our partitioning scheme. 
We choose $l = \lg^2 n/3$ and use Lemma~\ref{lem:facediv} to divide each internal face of $G$ that has more than $\lg^2 n$ vertices into smaller faces. 
We denote the resulted graph by $G'$. 
Thus any internal face of $G'$ has at most $\lg^2 n$ vertices. 
We call a face of $G'$ that is a face of $G$ an {\em original face}, and a face of $G'$ that is part of a larger face of $G$ a {\em modified face}. 
By Lemma~\ref{lem:facediv}, among the modified faces of $G$' that are in the same face of $G$, there is at most one modified face that has less than $l$ vertices. 
Therefore, the total number of modified faces of $G'$ is $O(f / l) = O(n / \lg^2 n)$. 
Let $f'$ be the number of internal faces of $G'$. 
Then we have $ 4(2n-5)/\lg^2 n \le f' \le 2n-5$. 

For the top-level partition, we choose $t = \lg^8 f' / f'$ to apply Lemma~\ref{lem:tfacesub} on $G'$. 
Then the $t$-face separator, $S$, of $G'$ has $O(\sqrt{f'/t'}) = O(f' / \lg^4 f') = O(n/\lg^4 n)$ faces. 
As each face of $G'$ has at most $\lg^2 n$ vertices, the number of vertices in $S$ is at most $O(n/\lg^2 n)$. 
We call each adjacent face component of $G'\setminus S$ a {\em region}. 
By Lemma~\ref{lem:componentnosub}, there are at most $O(n/\lg^2 n)$ regions. 
The number of faces of each region is at most $tf' = \lg^8 f' = O(\lg^8 n)$, so each region has at most $O(\lg^{10} n)$ vertices. 
By Lemma~\ref{lem:dupsub}, the sum of the duplication degrees of the boundary vertices of all the regions is $O(n/\lg^2 n)$. 

Consider a region $R_i$. 
Let $f_i$ and $n_i$ be the number of faces and vertices of $R_i$, respectively. 
Then $f_i = O(\lg^8 n)$ and $n_i =O(\lg^{10} n)$. 
We choose $l = \lg^{1/4} n/3$ to apply Lemma~\ref{lem:facediv} to divide each internal face of $R_i$ that has more than $\lg^{1/4} n$ vertices into smaller faces. 
We denote the resulted graph by $R_i'$. 
Thus any internal face of $R_i'$ has at most $\lg^{1/4} n$ vertices. 
We call a face of $R_i'$ that is a face of $R_i$ an {\em original region face}, and a face of $R_i'$ that is part of a larger face of $R_i$ a {\em modified region face}. 
Same as the analysis for $G'$, we have that the number of modified region faces in $R_i$ is $O(n_i/\lg^{1/4} n)$, so the total number of modified region faces in all the regions of $G'$ is $O(n/\lg^{1/4} n)$. 
Let $f_i'$ be the number of internal faces of $R_i'$. 
We also have $4(2n_i-5)/\lg^{1/4} n \le f_i' \le 2n_i-5$. 

We perform bottom-level partition on each region $R_i$. 
Let $t_i=\lg^{3/4}n / f_i'$. 
We use Lemma~\ref{lem:tfacesub} to compute a $t_i$-face separator, $S_i$, for $R_i'$. 
Then $S_i$ has $O(\sqrt{f_i'/t_i}) = O(f_i'/\lg^{3/8}n)$ faces, so $S_i$ has $O(f_i'/\lg^{1/8} n)$ vertices. 
We call each adjacent face component of $R_i'\setminus S_i$ a {\em subregion} of $R_i'$ (or $R_i$), and we denote the $j${\kth} subregion of $R_i'$ by $R_{i,j}$. 
By Lemma~\ref{lem:componentnosub}, for region $R_i$, there are at most $O(\lg^{1/4} n \times \sqrt{f_i'/t_i}) = O(f_i'/\lg^{1/8} n)$ subregions. 
As $f_i' \le 2n_i-5$, the total number of subregions of all the regions of $G_i$ is $O(n/\lg^{1/8} n)$. 
The number of faces of each region is at most $t_if_i' = \lg^{3/4} n$, so each region has at most $\lg n$ vertices. 
By Lemma~\ref{lem:dupsub}, the sum of the duplication degrees of the boundary vertices of all the subregions of $R_i$ is $O(n_i/\lg^{1/8} n)$, so the sum of the duplication degrees of all the boundary vertices of the subregions in the entire graph $G$ is $O(n/\lg^{1/8} n)$. 

\subsection{Vertex Labels and Face Labels}
\label{sec:labelsub}

We now design a labeling scheme for the vertices based on the two-level partition in Section~\ref{sec:twolevelsub}. 
Same as the case of planar triangulations, we assign a distinct number called {\em graph-label} from the set $[n]$ to each vertex $x$ of $G$. 
Each vertex $x$ also has a {\em region-label} for each region $R_i$ it is in, which is a distinct number from the set $[n_i]$. 
The {\em subregion-label} of $x$ is defined similarly at the subregion level. 

We use the techniques in Section~\ref{sec:label} to assign the labels from bottom up. 
To assign the subregion-labels, 
given a subregion $R_{i,j}$, 
we use Lemma~\ref{lem:encodesub} to permute its vertices. 
If a vertex $x$ in $R_{i,j}$ is the $k${\kth} vertex in this permutation, then the subregion-label of $x$ in $R_{i,j}$ is $k$. 
With the subregion labels assigned, we use exactly the same process in Section~\ref{sec:label} to compute the region-labels and graph-labels of all the vertices. 

As we use the same technique in Section~\ref{sec:label} to label the vertices (except that we use Lemma~\ref{lem:encodesub} instead of Lemma~\ref{lem:encodetriag}), the techniques of Lemma~\ref{lem:labconversion} can be used to perform constant-time conversions from subregion-labels (or region-labels) to region-labels (or graph-labels). 
The analysis of the number of regions/subregions and the sums of the duplication degrees of the boundary vertices of the regions/subregions of $G$ in Section~\ref{sec:twolevelsub} guarantees that the space required is still $o(n)$ bits. 
Thus we have the following lemma:
\begin{lemma}
\label{lem:labconversionsub}
There is a data structure of $o(n)$ bits such that given a vertex $x$ as a subregion-label $k$ in subregion $R_{i,j}$, the region-label of $x$ in $R_i$ can be computed in $O(1)$ time. 
Similarly, there is a data structure of $o(n)$ bits such that given a vertex $x$ as a region-label $k$ in region $R_i$, the graph-label of $x$ can be computed in $O(1)$ time. 
\end{lemma}

For planar subdivisions, we also need design a labeling schemes for its faces. 
We do not have to do this for planar triangulations, because in that case, each face has three vertices, and it is sufficient to locate these three vertices to return the face. 
However, we cannot do so for general planar subdivisions, because a face may have a large number of vertices, and it may take too much time to return all these vertices. 

For each face of $G$, we assign a distinct number called {\em graph-id} from the set $[f]$. 
This is the identifier we return when answering point location queries. 
For each face in a region $R_i$, we also assign a distinct number called {\em region-id} from the set $[f_i]$. 
Note that a face of the region $R_i$ is not necessarily a face of $G$ (i.e. it can be a modified face instead of an original face). 
For each face in a subregion $R_{i,j}$, we assign a distinct number called {\em subregion-id} from the set $[f_{i,j}]$, where $f_{i,j}$ is the number of faces in $R_{i,j}$. 
Again a face of $R_{i,j}$ is not necessarily a face of $R_i$. 

We number the faces from bottom up. 
For each subregion $R_{i,j}$, we list its faces by a canonical order (such as BFS order) of the corresponding vertices of its dual graph. 
The $k${\kth} face listed is assigned $k$ as its subregion-id in $R_{i,j}$. 
To assign identifiers to a face $x$ in region $R_i$, there are two cases. 
First, we consider the case where $x$ is in at least one subregion of $R_i$, or part of it is a modified region face that is in at least one subregion of $R_i$. 
Let $q_i$ be the number of such faces. 
We assign a distinct number from $[q_i]$ to each such face as its region-id by computing a permuted sequence of all these faces as follows. 
We visit each subregion $R_{i,j}$, for $i=1,2,\cdots, u_i$, where $u_i$ is the number of subregions in $R_i$. 
When we visit $R_{i,j}$, we list all its faces sorted by their subregion-ids in increasing order. 
As some faces of $R_{i,j}$ are modified region faces, we replace these modified region faces by the faces of $R_i$ that they are in. 
This way after we visit all the subregions of $R_{i,j}$, we have a sequence of the faces of $R_i$ that are in this case. 
Note that each face of $R_i$ may occur multiple times in this sequence, and by only keeping its first occurrence in the sequence, we have a permuted sequence of such faces. 
We assign number $k$ to the $k${\kth} face of $R_i$ as its region-id. 
Second, for the case where $x$ or parts of it only exist in $S_i$, we arbitrarily assign a distinct number from the set $[q_i, q_{i+1}, \cdots, f_i]$ to each such face as its region-id. 
The approach to assign graph-ids to the faces of $G$ is similar, except that we perform the above process for the top-level partition. 

Given a face of a subregion (or region) and its subregion-id (or region-id), we need find the identifier of the face in the corresponding region (or in $G$) that contains this face. 
To do this, we have the following lemma: 
\begin{lemma}
\label{lem:faceconversion}
There is a data structure of $o(n)$ bits such that given a face $x$ with subregion-id $k$ in subregion $R_{i,j}$, the region-label of the face in $R_i$ that contains $x$ can be computed in $O(1)$ time. 
Similarly, there is a data structure of $o(n)$ bits such that given a face $x$ with region-id $k$ in region $R_i$, the graph-id of the face of $G$ that contains $x$ can be computed in $O(1)$ time. 
\end{lemma}
\begin{proof}
We only show how to prove the first part of this lemma; the second part can be proved similarly. 

We use the same notation as in the previous part of this Section. 
As only faces of $R_i$ that have more than $\lg^2 n$ vertices are divided into modified region faces, the number of modified region faces of $R_i'$ is $O(n_i/\lg^2 n)$. 
Thus $f_i' = f_i + O(n_i/\lg^2 n)$. 

Observe that when we compute the region-ids of the faces of $R_i$, we construct a conceptual array of length $f_i'$ (i.e. the sequence we obtain before we remove the multiple occurrences of faces in it to obtain the permuted sequence of the faces of $R_i$). 
We denote this array by $E_i$. 
$E_i$ have the answers of our queries, but we cannot afford storing it explicitly. 
Instead, we construct the following data structures for each region $R_i$:

\begin{itemize}
\item A bit vector $F_i[1..f_i']$ which stores the numbers $f_{i,1}, f_{i,2}, \cdots, f_{i, u_i}$ in unary, i.e. $F_i = 10^{f_{i,1}-1}10^{f_{i,2}-1}\cdots10^{f_{i,u_i}-1}$;
\item A bit vector $J_i[1..f_i']$ in which $J_i[k] = 1$ iff the first occurrence of the region-id $E[k]$ in $E_i$ is at position $k$ (let $z_i$ denotes the number of $0$s in $J_i$);
\item An array $K_i[1..z_i]$ in which $K_i[k]$ stores the region-id of the face that corresponds to the $k${\kth} $0$ in $J_i$, i.e. $K_i[k] = E_i[\selop_{J_i}(0, k)]$.
\end{itemize}

The analysis of the space costs of these data structures is similar to that in the proof of Lemma~\ref{lem:labconversion}. 
We can prove that space costs is $o(n)$ bits. 
The same algorithm in Lemma~\ref{lem:labconversion} can be used to compute the region-id of $x$ in constant time. \qed
\end{proof}

\subsection{Supporting Point Location Queries}
\begin{theorem}
\label{thm:pointlocationsub}
Given a planar subdivision $G$ of $n$ vertices, 
there is a succinct geometric index of $o(n)$ bits that supports point location on $G$ in $O(\lg n)$ time. 
This index can be constructed in $O(n)$ time. 
\end{theorem}
\begin{proof}
We preform the two-level partitioning of $G$ as in Section~\ref{sec:twolevelsub}, and use the approaches in Section~\ref{sec:labelsub} to assign labels to the vertices and faces of $G$, but we do not store these labels explicitly. 
Instead, we sort the vertices by their graph-labels in increasing order, and store their coordinates as a sequence. 
We then show how to construct a succinct geometric index of $o(n)$ bits.

The succinct geometric index consists of three sets of data structures. 
This first set of data structures are the $o(n)$-bit data structures constructed in Lemmas~\ref{lem:labconversionsub} and \ref{lem:faceconversion} that supports conversions between subregion-labels (or subregion-ids), region-labels (or region-ids) and graph-labels (or graph-ids). 
The second and the third sets of data structures correspond to the top-level and the bottom-level partitions. 

To construct the data structures for the top-level partition, we consider the graph $S'$ that can be constructed by triangulating the graph consisting of the separator $S$ and the outer face of $G$. 
$S'$ is a planar triangulation of $O(n / \lg^2 n)$ vertices, so we can use the approach of Kirkpatrick~\cite{ki1983} to construct a data structure $P$ of $O(n/\lg^2 n)$ words (i.e. $O(n/ \lg n)$ bits) to support the point location queries on $S'$. 
Note that when we construct $P$, we simply use the graph-label of any vertex to retrieve its coordinates, so that we do not store any coordinate in $P$. 
For each face of $S'$, we store an integer and a bit. If this face is in region $R_i$, we store $i$ and a bit $1$. 
If it is a face in separator $S$, we store the graph-id of the face of $G$ that contains this face and a bit $0$. 
As there are $O(n/\lg^2 n)$ faces and the number assigned to each face can be stored in $\lg n$ bits, the space required to store these numbers and bits is $O(n/ \lg n)$ bits. 
All these ($P$ and the values assigned to the faces of $S'$) are the data structures for the top-level partition of $G$, and they occupy $O(n/ \lg n)$ bits. 

The data structures for the bottom-level partition are constructed over the regions of $G$. 
Given a region $R_i$,  we consider the graph $S_i'$ that can be constructed by triangulating the graph consisting of the separator $S_i$ and the outer face of $G$. 
Then $S_i'$ has $O(n_i / \lg^{1/8} n)$ vertices, so a pointer that refers to a vertex of $S_i'$ can be stored in $O(\lg\lg n)$ bits. 
To refer to the coordinates of any vertex of $S_i'$, we only uses its region-label as we can compute its graph-label in constant time by Lemma~\ref{lem:labconversionsub}, and $O(\lg\lg n)$ bits are sufficient to store a region-label. 
Thus we can use the approach of Kirkpatrick~\cite{ki1983} to construct a data structure $P_i$ of $O(n_i\lg\lg n/\lg^{1/8} n)$ bits to support the point location queries on $S_i'$ in $O(\lg\lg n)$ time. 
We store a number for each face of $S_i'$, and this number is $j$ if this face is in subregion $R_{i,j}$, and if it is a face in separator $S_i$, we explicitly store the region-id of the face containing it. 
We also use a bit to indicate whether a face of $S_i'$ is in a subregion or not. 
The space cost of storing these numbers and bits is $O(n_i\lg\lg n/\lg^{1/8} n)$ bits. 
The space cost of these data structures for all the regions is $O(n\lg\lg n / \lg^{1/8} n) = o(n)$ bits. 

Therefore, the succinct geometric index constructed above occupies $o(n)$ bits. 

To support point location queries, 
given a query point $x$, we first use the set of data structures constructed for the top-level partition to perform a point location query on the graph $S'$ in $O(\lg n)$ time using $x$ as the query point. 
This tells us whether $x$ is in a face of $S$ or not. 
If it is, we also have the graph-id of this face and we return it as the result. 
Otherwise, we get the number of the region that $x$ is in. 
Assume that $x$ is in region $R_i$. 
We then use the bottom-level data structures to perform a point location query on the graph $S_i'$ in $O(\lg\lg n)$ time using $x$ as the query point. 
Similarly, this tells us whether $x$ is in a face of $S_i$ or not. 
If it is, we also have the region-id of this face, and we compute its graph-id using Lemma~\ref{lem:faceconversion} in constant time and return it as the result. 
Otherwise, we get the subregion that $x$ is in. 
Assume that it is $R_{i,j}$. 
As each subregion has at most $\lg n$ points, and by Lemma~\ref{lem:labconversionsub}, we can compute the graph-label of any of these vertices in constant time, and thus retrieve its coordinates in $O(1)$ time, we can use Lemma~\ref{lem:encodesub} to construct the graph $R_{i,j}'$ in $O(\lg n)$ time, and then check each of its faces to find out the face that $x$ is in. 
The subregion-id of this face can be determined from the dual graph, and can be used to compute its graph-id in constant time. 
Therefore, the entire process takes $O(\lg n)$ time. 

Same as the analysis in Lemma~\ref{lem:plpre}, the preprocessing time is $O(n)$. 
\qed
\end{proof}

We can use this theorem to solve the following problem. 
Given a simple polygon and a query point, we want to test whether the
polygon contains the query point. 
This is called the {\em membership query} on
the polygon.
\begin{corollary}
Given a simple polygon of $n$ vertices, 
there is a succinct geometric index of $o(n)$ bits that supports membership query on the polygon in $O(\lg n)$ time. 
This index can be constructed in $O(n)$ time. 
%
\end{corollary}
\begin{proof}
We simply choose an orthogonal rectangle in the plane that contains this polygon. 
This rectangle and the polygon form a planar subdivision $G$ of two faces. 
We apply Theorem~\ref{thm:pointlocationsub} to $G$. 
Given a query point $x$, we can check whether it is inside or outside the rectangle in constant time. 
If it is inside the rectangle, we further use the succinct geometric index for $G$ to decide which face it is in. 
\qed
\end{proof}

\section{Applications}

\begin{REMOVED}
\subsection{Membership Query in A Simple Polygon.}
Given a simple polygon and a query point, to test whether the
polygon contains the query point is the {\em membership query} on
the polygon.
\begin{corollary}
Consider a simple polygon of $n$ vertices. With a permutation of
these vertices and a succinct geometric index of $O(n)$ bits, any
membership query on the polygon can be answered in $O(\lg n)$ time
with $O(\lg n)$ space. Such a permutation and the succinct index
can computed in $O(n)$ time.
\end{corollary}
\begin{proof}
We first triangulate the polygon and its outer face in $O(n)$
time~\cite{ch1991}. Using Theorem~\ref{thm:pointlocation}, we can
compute a permutation of the vertices and a succinct geometric
index of size $o(n)$. Then we modify this succinct index to
identify whether each triangle in the triangulation is inside the
polygon. Consider the two-level partitioning scheme used in
Theorem~\ref{thm:pointlocation}. Additional triangulations are
built for the top-level and the bottom-level partitions. In the
succinct index for these triangulations, we assign one extra bit
for each triangle, to identify whether it is inside the polygon.
There are two possible cases: the triangle is completely inside
the polygon or the triangle is not completely the polygon. This
modification does not increase the space bound of the original
data structure. For each subregion, the permutation of the
vertices represents the order to construct the triangulation. We
attach a bit vector to indicate whether each triangle, in this
order, is inside the polygon. The total size of bit vectors for
all subregions is $O(n)$. With the modified succinct index, we can
locate the triangle containing the query point and identify
whether this triangle is inside the polygon in $O(\lg n)$ time.
\qed
\end{proof}

Same as Corollary~\ref{cor:integer}, we also can answer point location queries in
trapezoidal decompositions, and arbitrary planar subdivisions in
sub-logarithmic time, if all coordinates are integers.
\end{REMOVED}

\subsection{Vertical Ray Shooting Query}

Given a set of disjoint line segments, we can build its
trapezoidal decomposition in $O(n\lg n)$ time~\cite{DVOS97}. By
answering point-location query in a trapezoidal decomposition, the
trapezoid containing the query point is reported. Alternatively,
we can return the line segments defining the upper and the lower
edges of the trapezoid. This query is referred as the {\em
vertical ray shooting query.}
\begin{theorem}
Given a set of disjoint line segments in the plane, 
there is a succinct geometric index of $o(n)$ bits that supports vertical ray shooting on this set in $O(\lg n)$ time. 
This index can be constructed in $O(n\lg n)$ time. 
\end{theorem}
\begin{proof}
We first build the trapezoidal decomposition of the plane with the
set of given line segments, which takes $O(n\log n)$
time~\cite{DVOS97}. In this planar subdivision, each face is a
trapezoid; each vertex is determined by at most two line segments.

Similar to building the succinct index for planar triangulations,
we apply the two-level partitioning scheme on the trapezoidal
decomposition. In the top-level partition, $O(n/\lg^3 n)$
trapezoids are selected. The line segments defining selected
trapezoids are grouped together. The rest of line segments are
grouped based on their regions. In the bottom-level partitions,
$O(n/\log^{1/2} n)$ trapezoids are selected in total. Again, line
segments are grouped together for the separator and each
subregions. Different from the triangular subdivision, each vertex
is defined by at most two line segments. In the data structure, it
is represented by two labels and a flag. Similar to the extra data
structure used in the triangular subdivision, we also build point
location structures for both levels and this index only takes
$o(n)$ bits in total.

Inside subregions, we handle differently. We simply group all line
segments in the same subregion in an arbitrary order. 
We can use a similar labeling scheme as in Section~\ref{sec:label}. 
This works because each time a line segment is in two different regions (or subregions), it must contain at least one vertex in the separator of the graph (or of the region). 
As no two given segments intersect, we can use this fact to bound the sum of the number of regions (or subregions) that a line segment can be in. 

When processing the query, we only use the additional data
structure to locate the trapezoid containing the query point in
the top-level and bottom-level partitions, and also locate the
subregion containing the query point. In the subregion, we scan
through all line segments, and determine the line segments above
(below) the query point. Comparing with the line-segments above
(below) the query point in two levels, we obtain the real answer.
\qed
\end{proof}

\subsection{Implicit Geometric Data Structures}
Implicit geometric data structures have been
studied since 2000~\cite{BMM07,BCC04,BI04,CC08}. For example,
several implicit 2-d nearest neighbor query structures 
(this is equivalent to supporting point location on Voronoi Diagrams) 
have
been proposed~\cite{BCC04, CC08}. However, implicit point location
structures for planar subdivisions are still unknown. We
can apply succinct geometric indexes to solve this problem.

We apply the well-known bit encoding technique used in many
previous work on in-place algorithms and implicit data structures
(e.g.~\cite{M86}): divide the array of input into consecutive
pairs. We permute each pair of the data. In the lexical order, if
the first datum is smaller, 0 is encoded, otherwise 1 is encode.
(Assume we have removed all duplicates.) Retrieving one bit in
this encoded data structure requires $O(1)$ time, and retrieving a
pointer of size $O(\lg n)$ requires $O(\lg n)$ time.

Before encoding the succinct data structure, in the input array,
we put each pair of the data in the lexical increasing order. To
ensure the data structure is still valid above the above
permutation, when applying the separator theorem, we need to ensure
all regions (and subregions) have an even size. If one region has
an odd size, we can move one of the element in that region into
the separator. When computing the sequence to construct the
triangulation of a sub-region, we ensure the number of vertices
selected in one round in~\cite{ds1997} is even, and every time two
vertices are removed from the triangulation together.

In the succinct index, we store labels based on the sequence
described above. By permuting each pair of data, we can encode
$1/2n$ bits in the array. Therefore, we encode the succinct data
structure of size $o(n)$ in the input array, for a $n$ large
enough. For point location queries, the input array is a set of
vertices, and for vertical ray shooting queries, the input array is a set
of line segments.

When we answer queries, we locate the pair of data in the array,
put them in the lexical increasing order, then locate the datum
based on the label stored in the succinct index. Then we have:

\begin{theorem}
Given a planar triangulation of $n$ vertices, there is a permutation
of the vertex coordinates array that can support point location
 in $O(\lg^2 n)$ with $O(1)$ words of working space.
\end{theorem}

Similarly, we can design implicit data structures to answer
point location queries in an arbitrary planar subdivision and support vertical ray shooting. These data structures are represented
as a permutation of the input array, and answering point location
queries in these data structures takes $O(\lg^2 n)$ time.

\section{Conclusion}
In this paper, we start a new line of research by designing succinct geometric indexes. 
We design a succinct geometric for triangular planar subdivision that occupies $o(n)$ bits that, by taking advantage of the points permuted and stored elsewhere as a sequence, to support point location in $O(\lg n)$ time. 
We also considered the exact number of point-line comparisons, integer coordinates from a bounded universe, and the case where the query has a certain distribution, by designing three succinct geometric indexes for them. 
We also generalize our techniques to planar subdivisions, and apply them to design succinct geometric indexes for vertical ray shooting and the design of implicit data structures supporting point location in $O(\lg^2 n)$ time. 
In addition, we believe that our techniques are practical. 
This is because several previous results we use have practical implementations, such as practical bit vectors~\cite{fggv2006}, and we can group constant number of subregions to have enough vertices in order to apply Lemma~\ref{lem:encodetriag}. 
Thus we expect our technique to influence the design of space-efficient geometric data structures. 

There are a few open problems. First, the index we design for the case where the query distribution is known supports point location in $O(H+1)$ expected time. 
Thus it is an open problem to improve this to $H + o(H)$ expected number of comparisons. 
Another open problem is to design succinct geometric indexes for other types of queries, such as general ray shooting.

\bibliographystyle{plain}
\bibliography{pointlocation}

\end{document}